\documentclass{article}

\usepackage[auth-lg]{authblk}
\usepackage{graphicx}
\usepackage{color}
\usepackage{times}
\usepackage{fullpage}
\usepackage{ifpdf}
\usepackage{color}
\usepackage{amsmath,amssymb}

\DeclareMathAlphabet{\mathpzc}{OT1}{pzc}{m}{it}

\newtheorem{thm}{Theorem}[section]
\newtheorem{cor}[thm]{Corollary}
\newtheorem{lem}[thm]{Lemma}

\numberwithin{equation}{section}

\newcommand{\A}{\mathcal{A}}
\newcommand{\El}{\mathcal{E}_{\{\ell\}}}
\newcommand{\opt}{\mathbf{OPT}}
\newcommand{\sq}{\hbox{\rlap{$\sqcap$}$\sqcup$}}
\newcommand{\qed}{\hspace*{\fill}\sq}
\newenvironment{proof}{\noindent {\bf Proof.}\ }{\qed\par\vskip 4mm\par}

\newcommand{\ccp}{{\tt CCP}}
\newcommand{\uccp}{{\tt UCCP}}
\newcommand{\spt}{{\tt SPT}}
\newcommand{\sptg}{{\tt SPT-G}}
\newcommand{\spc}{{\tt SPC}}
\newcommand{\sink}{{\mathbf{sink}}}
\newcommand{\vn}{{\mathpzc{n}}}
\newcommand{\vm}{{\mathpzc{m}}}

\begin{document}
\begin{titlepage}

\title{Energy-Efficient Shortest Path Algorithms for Convergecast in Sensor
Networks}

\author[1]{John Augustine}
\author[2]{Qi Han}
\author[2]{Philip Loden}
\author[1]{Sachin Lodha}
\author[1]{Sasanka Roy}
\affil[1]{Tata Research Development and Design Centre\\

54B, Hadapsar Industrial Estate, Hadapsar\\

Pune, India 411013 \\

{\tt \{john.augustine, sachin.lodha, sasanka.roy\}@tcs.com}\\

~}

\affil[2]{   Dept. of Math and Computer Sciences\\

Colorado School of Mines, Golden, CO 80401 \\

{\tt \{ploden,qhan\}@mines.edu}}
\date{}

\maketitle \thispagestyle{empty}

\begin{abstract}
We introduce a variant of the capacitated vehicle routing problem
that is encountered in sensor networks for scientific data
collection. Consider an undirected graph $G=(V \cup
\{\mathbf{sink}\},E)$. Each vertex $v \in V$ holds a constant-sized
reading normalized to $1$ byte that needs to be communicated to the
$\mathbf{sink}$. The communication protocol is defined such that
readings travel in packets. The packets have a capacity of $k$
bytes. We define a {\em packet hop} to be the communication of a
packet from a vertex to its neighbor. Each packet hop drains one
unit of energy and therefore, we need to communicate the readings to
the $\mathbf{sink}$ with the fewest number of hops.

We show this problem to be NP-hard and counter it with a simple
distributed $(2-\frac{3}{2k})$-approximation algorithm called {\tt
SPT}~that uses the shortest path tree rooted at the $\mathbf{sink}$.
We also show that {\tt SPT}~is absolutely optimal when $G$ is a tree
and asymptotically optimal when $G$ is a grid. Furthermore, {\tt
SPT}~has two nice properties. Firstly, the readings always travel
along a shortest path toward the $\mathbf{sink}$, which makes it an
appealing solution to the convergecast problem as it fits the
natural intuition. Secondly, each node employs a very elementary
packing strategy. Given all the readings that enter into the node,
it sends out as many fully packed packets as possible followed by at
most 1 partial packet. We show that any solution that has either one
of the two properties cannot be a $(2-\epsilon)$-approximation, for
any fixed $\epsilon > 0$. This makes \spt~optimal for the class of
algorithms that obey either one of those properties.
\end{abstract}

\bigskip

\centerline{{\bf Keywords}: Sensor Networks, Shortest Path, Vehicle
Routing, Bin Packing, Convergecast}

\end{titlepage}

\section{Introduction}

We introduce a problem that combines
bin-packing~\cite{CofGarJoh97,Garey1979} and network routing. To
motivate the problem, we develop it as a variant of the capacitated
vehicle routing problem. Going beyond theoretical interest, we also
show how this problem arises naturally in the study of sensor
networks that collect precise data about the environment with
minimal amount of energy expenditure.

Consider a basic version of the capacitated vehicle routing problem
in which there is a depot (or sink) and several agricultural towns
connected by a road network. Each town has some produce (in a sealed
bag or container) that is to be sent to the depot. A vehicle has to
pick up all the items and transport it to the depot. In so doing,
the vehicle must never carry items that exceed its load capacity and
the bags are not to be opened en route. Note that this problem
naturally combines bin packing and the traveling salesperson
problem. Our problem is a simplification, which retains the bin
packing aspect, but simplifies the transportation issues as
described in the following scenario.

Suppose we contract with a transportation company that provides us
trucks for transporting the sealed bags of produce from these towns
to a central hub. For simplicity, assume that trucks will be
available when  needed. The cost of a truck trip between any pair of
neighboring towns is a fixed price regardless of the distance
between them and how heavy (or light) the load is. However, we do not
pay for empty trips; the trucking company is responsible for moving
empty trucks to the next demand location. Our goal is to transport
all the produce to the depot, while minimizing the total cost of
hiring trucks from the trucking company.

Similarly, consider the convergecast problem in sensor networks. The
sensor nodes need to send sensed data to a centralized sink via
multiple hops. A sensor reading can usually be encoded in a few
bytes\footnote{We use bytes for simplicity, but any appropriate unit
of memory can be used.}, so more than one reading can fit into a
standard transmission packet, but there is a limit on the total
number of bytes that each packet can carry. Each reading has to stay
intact along the way. This is different from sensor data aggregation
where a function is performed over several sensor readings to,
typically, generate one single representative value for each region
being sensed~\cite{GE03}. While data aggregation is agreeable in
many situations, under certain scenarios, applications would rather
desire the collected data to be exact. This requirement is common in
scientific data gathering~\cite{PIL+09}. We have a cost associated
with each hop, which is independent of the number of readings in
it\footnote{This is an acceptable assumption commonly used in the
sensor network community, although more realistic radio model
indicates that packet size does matter~\cite{HCB00}.}. Consequently,
we ask the question: can we pack the readings in common routes to
minimize the number of hops? It is easy to see that the variant of
the capacitated vehicle routing and the sensor network convergecast
problem are equivalent. We define our problems drawing from the
terminology used in sensor networks.

~

\noindent{\bf Background Information:} Along with other Vehicle
Routing Problems (VRP), the capacitated variants of the VRP have
been studied since the 50's~\cite{DanRam59}. With the development of
complexity theory, it became clear that capacitated VRP combines two
different problems, the bin-packing problem and the traveling
salesperson problem, that are independently hard to
solve~\cite{Garey1979} in polynomial time. The capacitated VRP has
continued to receive attention over the decades. Charikar, Khuller
and Raghavachari \cite{ChaKhuRag01} provide a 5-approximation
algorithm for a multi-sink variant of the capacitated vehicle
routing problem that allows vehicles to drop off commodities at
intermediate locations. Readers interested in the history of
capacitated VRP and the various techniques used for solving them can
refer to the excellent book edited by Toth and Vigo~\cite{TotVig01}.
The techniques they cover include exact methods such as
branch-and-bound and branch-and-cut, and also consider set-covering
based algorithms. Additionally, they provide several heuristics and
meta-heuristics that work well in practice. To the best of our
knowledge, we have not encountered the exact variant we study in any
prior work.

The convergecast problem has obtained prominence among sensor
networks researchers because it fits well with the goal of sensor
networks, which is to monitor and collect data about an environment.
The focus has been to either minimize the time, the energy, or the
dual-criteria of both time and energy required to complete the
convergecast~\cite{
GanZhaHua07,HohDohBre04, LinRagSiv02,LuKriRag07,PanTse08,ParHan08, YuPra05,UpaGup07,ZhaGanHua07}.
Researchers have also exploited
spatial locality in many real-life convergecast scenarios by
aggregating the data and transmitting the representative values for
sub-regions within the region being sensed~\cite{GE03,KriEstWic02}.

~

\noindent{\bf Problem Definitions:} We are given a connected
 graph $G=(V \cup \{\sink\}, E)$ that is both
undirected and unweighted. An edge $e=(u,v)\in E$ implies that $u$
can communicate with $v$ and vice versa. Each vertex $v$ has a
single  reading of integral number of bytes $s(v)$ that has to be
reported to the appropriately denoted vertex $\sink$. These readings
must travel to the $\sink$ in packets that have a capacity of $k$
bytes. Since the readings have to fit in the packets, $\forall v \in
V$, $s(v) \leq k$. A packet consumes 1 unit of energy every time it
hops from a vertex to a neighbor regardless of the total size of the
readings in it. Our objective is to minimize the total energy
consumed to send all the readings to the $\sink$. We primarily seek
distributed routing algorithms in which the individual nodes are
unaware of the entire graph; they are only aware of their immediate
neighbors. We call this the {\tt Convergecast Problem} or the \ccp.

The \ccp~combines aspects of both bin-packing and routing. In
Theorem \ref{thm:ccpHard}, we show that it is NP-hard even when the
underlying graphs are restricted to a line or a tree of depth
greater than 1.  So, we limit our study to a simplification in which
the size of each reading is exactly 1 byte. We call this the {\tt
Unit Convergecast Problem} or the \uccp. In practice, many wireless
sensor applications such as room temperature monitoring for energy
conservation only need to deploy simple sensors with one single
sensing attribute.  These sensors then report small constant-sized
readings as directed. In our formulation, we normalize it to one
byte. The \uccp~helps us gain insight when the effect of bin-packing
is minimal because up to $k$ single-byte-sized readings can be
trivially placed into a packet. Interestingly, we show that even
\uccp~is NP-hard.

We now describe two desirable properties that we would like to see
in our solution.
\begin{description}
\item[Shortest Path Property:] An algorithm for \ccp~or \uccp~is said to  follow the
shortest path property if every packet hop always moves the packet
closer to the $\sink$. We refer to algorithms that have this
property as shortest path algorithms. Because we are concerned with
the convergecast problem, this property, when present, will make the
solution more intuitive. This is essentially geographic routing with
greedy forwarding often used in wireless sensor
networks~\cite{AY05}. Note also that even distributed networks, with
a little preprocessing, can easily establish a shortest path tree as
long as the graph is connected.
\item[Elementary Packing Property:] An algorithm for \uccp~is said to have the elementary
packing property if each vertex communicates at most one partial
packet and all the other packets, if any, are full. Such algorithms
are called elementary algorithms. An elementary algorithm ensures
that each node repackages the readings in the most straightforward
manner. It also ensures that communication overhead in the entire network
is minimized.  This is because minimal number of
packets will be used,  leading to minimal total number
of bytes in all the packets is minimal, since each packet has a constant-size
packet header.
\end{description}

In Section~\ref{sec:hard}, we prove that \ccp~is NP-hard even when
the underlying graph is very simple. We then shift our attention to
\uccp~and prove that it is also NP-hard. In Section~\ref{sec:spt},
we develop a $(2-\frac{3}{2k})$-approximation algorithm for \uccp.
It uses the shortest path tree in a very straightforward manner, and
hence, we call it the \spt.  Additionally, it is also an elementary
algorithm.  In Section~\ref{sec:inapprox}, we prove that any
algorithm that either follows the shortest path property or the
elementary packing property cannot guarantee a
$(2-\epsilon)$-approximation for \uccp. In light of this, if we
restrict ourselves to either shortest path algorithms or elementary
algorithms, then \spt~is optimal. In Section \ref{sec:special}, we
explore the performance of \spt~when the underlying graph is either
a tree or a grid and show that it is absolutely optimal in the
former case and asymptotically optimal in the latter case. Finally,
we discuss our experimental results in Section \ref{sec:expt}.

\section{Hardness Results}\label{sec:hard}

In this section we first show that \ccp~is NP-hard even for some of
the simplest trees via a reduction of SET-PARTITION to \ccp. This
result is formalized in Theorem \ref{thm:ccpHard}.

\begin{thm}\label{thm:ccpHard}
\ccp~is NP-hard even if the underlying graph $G$ is a straight line
or a tree of depth at least 2.
\end{thm}

\begin{proof}
Recall that in SET-PARTITION, we are given a set $U=\{x_1, x_2,
\cdots, x_\vn\}$ of integers. The question we ask is whether $U$ can
be partitioned into two subsets such that the sums of the elements
in either subsets are equal. SET-PARTITION is known to be
NP-complete~\cite{Garey1979}.

We can reduce an instance of SET-PARTITION to \ccp~in two very
simple ways as shown in Figure~\ref{fig:ccpNPhard}, which
illustrates the case when the instance of SET-PARTITION has 8
elements. We assume, without loss of generality, that the elements
of SET-PARTITION are  integral values between 1 and $k$ and add up
to $2k$. To reduce from SET-PARTITION to \ccp, we take each element
of the set and form an instance of \ccp~in which each element of $U$
forms a reading in \ccp~and is assigned to a node in \ccp.

\begin{figure}[!htb]
\centering \ifpdf
    \includegraphics[width=5in]{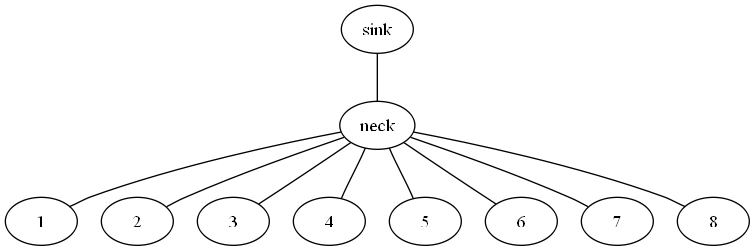} \\ \vspace{0.25in}
    \includegraphics[width=4in]{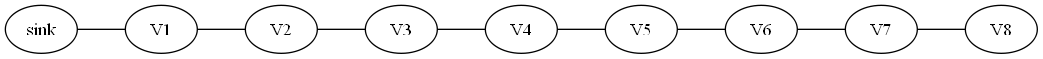}

\else
    \includegraphics[width=5in, bb=0 0 1043 59]{line.png}\\ \vspace{0.25in}
    \includegraphics[width=4in, bb=0 0 755 251]{treeLB.png}
\fi \caption{The two figures illustrate the reductions from
SET-PARTITION to \ccp.} \label{fig:ccpNPhard}
\end{figure}

In the case of the tree of depth 2, we include a ``neck" vertex
which is assigned a reading of size $k$. The other nodes with
readings assigned to them from SET-PARTITION are of degree 1 and are
connected to the neck. The neck is connected to the $\sink$. The
number of hops from the {\em neck} into the $\sink$ vertex will
depend on whether the SET-PARTITION instance can be partitioned into
two subsets.

Similarly, in the case of the line, the nodes form a linear chain
with one end connected to the $\sink$. Starting from the node
farthest away from the $\sink$, the readings travel toward the
$\sink$. At some point, there will be enough readings to require
exactly 2 packets for any reasonable algorithm. Note that the
$\sink$ has exactly one neighbour. Once all the readings reach that
neighbour, we will need either 2 or 3 packets to hop into the
$\sink$ depending on whether we can partition the set $U$ or not.
\end{proof}

We now turn our attention to \uccp. Interestingly, we show that even
\uccp~is NP-hard by reducing the set cover problem to it. In the
classic Set Cover Problem  we are given a ground set
$U=\{x_1,x_2,\ldots,x_{\vn}\}$ and a family of subsets
$S=\{S_1,S_2,\ldots,S_\vm\}$, $S_i \subseteq U$ for
$i=1,2,\ldots,\vm$. $C \subseteq  S$ is a cover if the union of
elements in $C$ is $U$. The goal is to find a cover $C_{min}
\subseteq S$ with the smallest cardinality. It is  well-known that
Set Cover Problem is NP-hard \cite{Garey1979}.

Given an instance of the set cover problem, we construct a sensor
network $T$ consisting of vertices arranged in three levels as
follows (refer Figure \ref{figure1}). Level 1 consists of only the
$\sink$ node. Level 2 nodes correspond to the
 sets $S_i \in S$ for $i=1,2,\ldots,\vm$. There is an edge
from each $S_i$ to $\sink$. We slightly abuse notation and use $S_i$
to also refer to the corresponding vertex. Level 3 consists of nodes
that correspond to  $\{x_1,x_2,\ldots, x_\vn\}$ which are the
elements of set $U$. Like level 2 nodes, we use $x_j$ to refer to a
level 3 vertex. Each node $x_j$ is connected by an edge to $S_i$ iff
the element $x_j \in S_i$ in the Set Cover instance.

We set the size of a packet to $k$ = $\max_i |S_i|$ bytes. We also
add another $k-1$ leaf nodes, which we call enforcers, to each
$S_i$. In Figure \ref{fig:ccpNPhard}, the enforcers are depicted by
a triangular pictorial gadget. Our objective is to solve the
convergecast problem for this setup of sensor networks. i,e. each
non-sink node (including the nodes in levels 2 and 3 and all the
enforcers) have a reading of 1 byte and we must pass each reading to
$\sink$ using the minimum number of packet hops.
\begin{figure}[!htb]
\centering
\ifpdf
    \includegraphics[width=70mm]{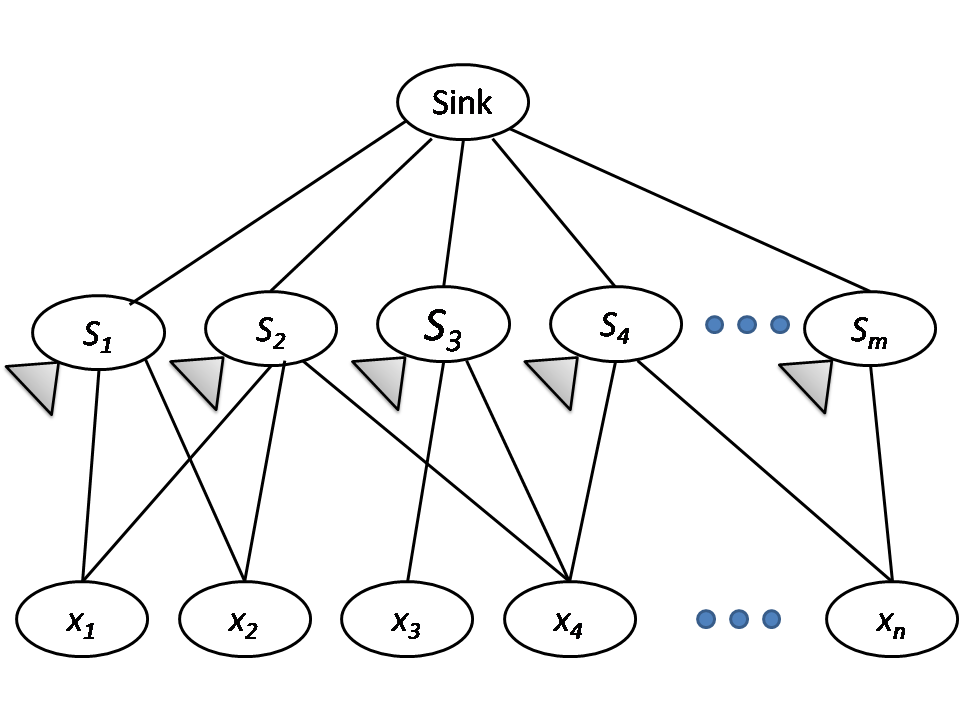} \hspace{0.5in}
    \includegraphics[width=45mm]{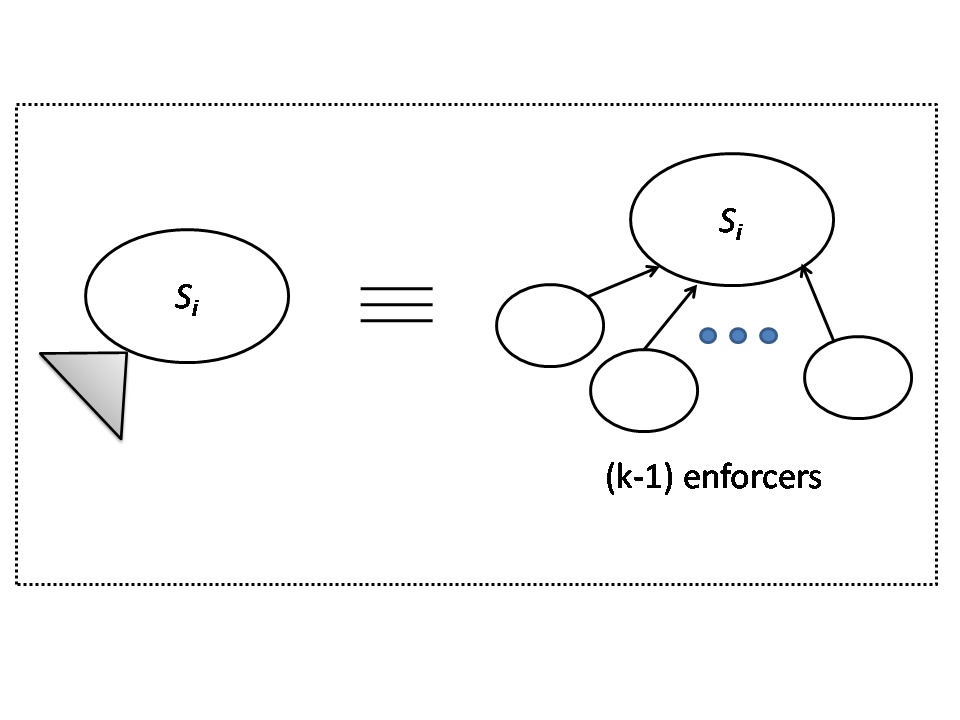}
\else
\includegraphics[width=70mm, bb=0 0 960 720]{NPHard.png} \hspace{0.5in}
    \includegraphics[width=45mm, bb=0 0 960 720]{Gadget.png}
\fi \caption{Reducing the Set Cover problem to the \uccp. The
enforcers are depicted as a triangular pictorial gadget; the actual
construction of the enforcers is shown in the box.} \label{figure1}
\end{figure}

For $K>0$, we can show that $\vn+\vm k + K$ hops suffice to route
each reading to the $\sink$ iff there exists a set cover of size
less than or equal to $K$ in the set cover problem. Each level 3
vertex has to send a packet to $\sink$ through a level 2 vertex.
Note that at least $\vn$ packets must hop out of the level 3
vertices for any solution (optimal or suboptimal). Consider the
portion of the graph consisting of a single level 2 node $S_i$, its
$k-1$ enforcers and $\sink$. Regardless of the activity outside this
portion, any solution requires $k$ hops because the $k-1$ enforcers
must communicate to $S_i$ and we need a packet from $S_i$ to the
$\sink$.
 Since there are $\vm$
such level 2 vertices, the number of hops is at least $\vm k$. If at
least one reading from level 3 vertex will hop through $S_i$, it
will force $S_i$ to
 send one more packet, which we call a {\em critical hop}. If $K$ is
 the number of critical hops, then we can cover the ground set by selecting the subsets
 corresponding to each of the $K$ chosen subsets. Therefore, the
 following theorem follows.
\begin{thm}
\uccp~is NP-hard.
\end{thm}
\section{The Shortest Path Tree (\spt) Algorithm} \label{sec:spt}

In this section, we present an algorithm that we call the Shortest
Path Tree Algorithm, or \spt, because it builds a shortest path tree
and only uses the edges in that tree. It is arguably the simplest
algorithm that uses the shortest path and follows elementary
packing. Therefore, it lends itself naturally to a distributed
implementation.

The steps in the \spt~algorithm are as follows. In the preliminary
phase, we find a shortest path tree $T$ of graph G rooted at
$\sink$. As a consequence, each node is aware of its parent and
children. Subsequently, each vertex waits till it has received all
packets from its children in $T$. Full packets are sent to its
parent as is. All the partial packets are re-packaged into the
maximum number of full packets and at most one partial packet and
all these packets are sent to the parent.

Let $T$ denote a shortest path tree rooted at $\sink$. Now we will
devise an algorithm that will use only the edges of $T$ to send the
reading of each node to $\sink$. Let $\opt$ and $\A$ be the number
of hops taken by the optimum solution and our algorithm,
respectively, in solving an instance of the \uccp. We  show that
${\A} \leq (2-\frac{3}{2k}) \times \opt$.

The maximum number of readings that can be packed in a packet is
$k$. If a packet contains $k$ readings then we call it a {\it full
packet}; otherwise, it is a {\it partial packet}. If a full packet
hops from a node $a$ to a neighbouring node $b$ then we will term
this as {\it full hop}. A partial hop is defined likewise.  We split
$\opt$ into $\opt^f$ and $\opt^p$ such that they are the number of
full and partial hops, respectively. We define $\A^f$ and $\A^p$ in
like manner. Naturally,
\begin{eqnarray}
  \opt &=& \opt^f + \opt^p \label{eqn:opt}\\
  \A &=& \A^f + \A^p. \label{eqn:spt}
\end{eqnarray}
Let us define the depth $d(v)$ of a node $v$ as the shortest
distance of a node $v$ from $\sink$ in $T$, i.e., the minimum number
of hops required for a reading to reach $\sink$ from $v$.


\begin{lem}
For any instance of the \uccp, $\A^p \leq 2 \cdot \opt^p$.
\label{lem:partials}
\end{lem}
\begin{proof}
Consider the packets that flow through a single vertex $v$ according
to any algorithm regardless of optimality. There is at least one
partial hop either out of $v$ or into $v$. We can prove this by
contradiction. Suppose there were no partial hops into $v$, but
$\ell$ full hops into $v$. Then, $k\cdot \ell +1$ readings would
have to hop out of $v$, which requires at least one partial hop.
This implies that at least $n/2$ hops are partial even for an
optimal algorithm. Therefore,
\begin{equation}\label{eqn:partialopt}
    \opt^p \geq n/2.
\end{equation}
According to the \spt~algorithm, each vertex waits for all its
children to communicate their packets and reorganizes the readings
such that at most one packet is not full. Therefore, $\A^p \leq n$,
 which, along with Equation \ref{eqn:partialopt}, completes the proof.
\end{proof}

Before we proceed into proving our theorem, we point out an obvious
property (formalized in Lemma~\ref{lem:readingMovement}) of any
algorithm that obeys the shortest path property, the \spt~being one
such algorithm. The reading corresponding to each vertex $v$ travels
a distance of exactly $d(v)$, which is the shortest distance to
reach the $\sink$. Therefore, the sum of all the distances traveled
taken over all {\em readings} (not packets) by \spt~is less than or
equal to that of any other algorithm. That sum is at least $\A^p + k
\A^f$ for \spt; we pessimistically account only one reading to have
hopped in each partial packet. Similarly, the sum of the distance
moved by {\em readings} according to an optimal algorithm is at most
$(k-1) \opt^p + k \opt^f$; we liberally account for $k-1$ readings
in each partial hop. Therefore, we can state the property as
follows:

\begin{lem}\label{lem:readingMovement}
For any instance of the \spt, $ \A^p + k \A^f \leq (k-1) \opt^p + k
\opt^f$.
\end{lem}
\begin{thm}\label{thm:uccp}
For any instance of~\uccp, $\A \leq (2-\frac{3}{2k}) \opt$.
\end{thm}
\begin{proof}
Using Equations \ref{eqn:opt} and \ref{eqn:spt}, we rewrite the
equation in Lemma \ref{lem:readingMovement} as
\begin{eqnarray*}
  k \A &\leq& (k-1) \opt + \opt^f + (k-1) \A^p \\
   &\leq& (k-1) \opt + \opt^f + \opt^p +  (k-1) \A^p - \A^p/2
                \quad \mbox{(using Lemma \ref{lem:partials})}\\
   &=& k \cdot \opt + (k-3/2) \A^p.
\end{eqnarray*}
Recall that $\A^p \leq n$. Hence, we can replace $\A^p$ with $\opt$
because $\opt$ is at least $n$; every vertex has to send out at
least one packet. Further, dividing by $k$ on both sides, we get $\A
\leq (2-\frac{3}{2k}) \cdot \opt$.
\end{proof}
Theorem~\ref{thm:uccp} proves the upper-bound for \spt, but the
underlying lemmas, Lemma~\ref{lem:partials} and
Lemma~\ref{lem:readingMovement}, are true for larger classes of
algorithms. Lemma~\ref{lem:partials} hold for any algorithm that
packs its readings in an elementary manner and
Lemma~\ref{lem:readingMovement} is true for any algorithm that
respects the shortest path property. Therefore we can state:
\begin{cor}\label{cor:ESP}
The approximation ratio of any algorithm in the class of algorithms
that obey the shortest path property and the elementary packing
property is at most $(2-\frac{3}{2k})$.
\end{cor}

Note that in \spt, each node sends its packets to one of its
parents. In practice, we might not want to burden one parent. This
can be alleviated by choosing a parent randomly. Alternatively, the
node can also choose a parent in a round-robin fashion. Corollary
\ref{cor:ESP} ensures that such variants will not incur a higher
hop-count than \spt. This can be of use to systems designers who are
interested in balancing the network overhead across the network
without compromising the hop-count.

\section{Lower Bounds on Approximating \uccp}
\label{sec:inapprox} Given the upper-bound on the approximation
ratio of \spt~in Theorem~\ref{thm:uccp}, a natural question we ask
is whether the analysis can be tightened. We are, however,
interested in algorithms that use shortest paths and employ
elementary packing. In this subsection, we discuss the
inapproximability of \uccp~when either one of those two properties
must be respected.
\begin{figure}[!htb]
\centering \ifpdf
    \includegraphics[width=120mm]{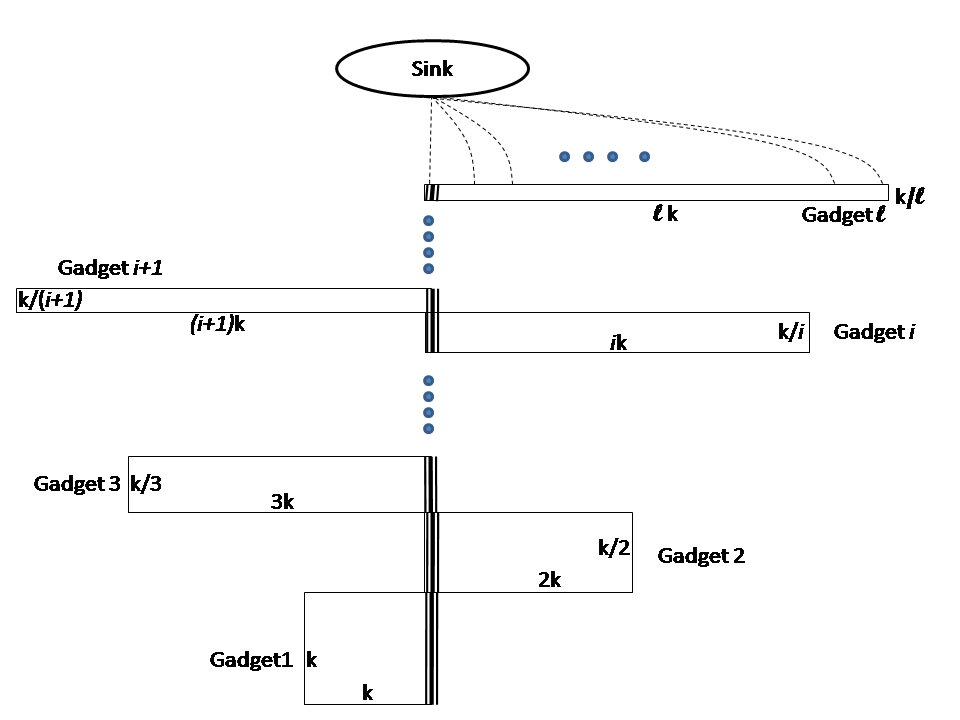}
\else
    \includegraphics[width=120mm, bb=0 0 960 720]{LBConst.png}
\fi \caption{The construction of an instance of \uccp~used to prove
Theorem~\ref{thm:inapproxSP} and Theorem~\ref{thm:inapproxElem}.
Note that the boxes are gadgets shown in Figure \ref{fig:gadget}.}
\label{fig:const}
\end{figure}

\begin{figure}[!htb]
\centering \ifpdf
    \includegraphics[width=130mm]{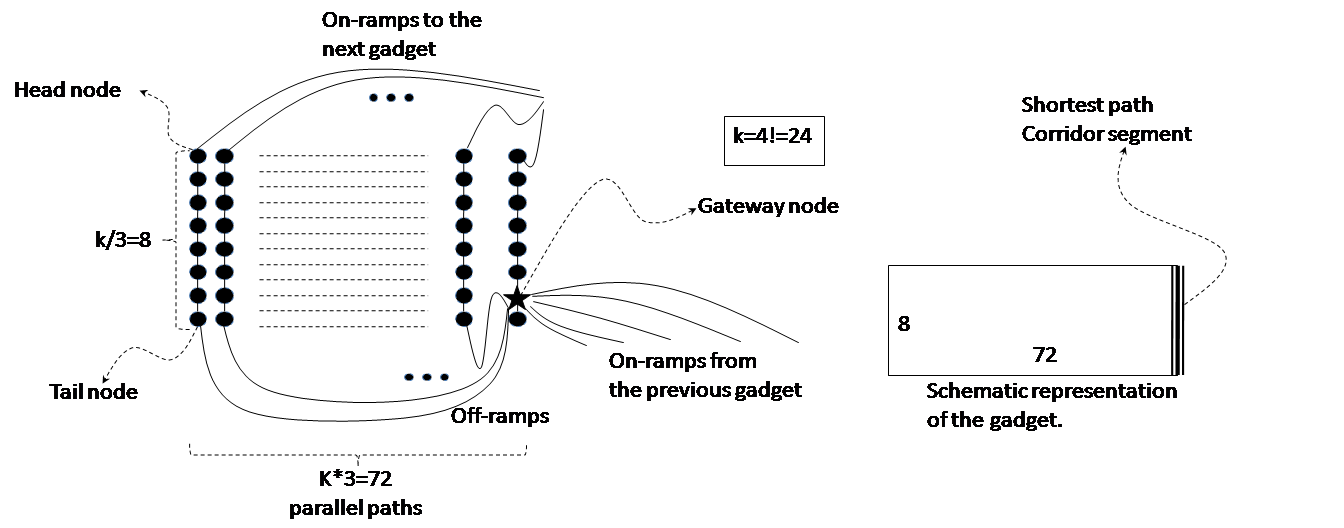}
\else
    \includegraphics[width=150mm, bb=0 0 1344 528]{LBGadget.png}
\fi \caption{A gadget for constructing the instance $\El$ of \uccp.
The schematic representation of the actual gadget is also provided
in the bottom right.} \label{fig:gadget}
\end{figure}

We begin by describing the construction of an instance $\El$ of the
\uccp, where $\ell$ is a positive integer. This instance is
constructed with one bad path (called the shortest path corridor or
\spc) to the $\sink$ such that an optimal algorithm can avoid it to
minimize the number of hops. However, in the construction, we ensure
that an algorithm that does not compromise on either the shortest
path property or the elementary packing property cannot avoid the
\spc~and therefore must hop more.

The instance $\El$ will consist of $\ell$ gadgets (shown in
Figure~\ref{fig:const}). The gadgets are indexed by $i$, $1\leq i
\leq \ell$. Gadget 1 is farthest away from the $\sink$ and gadget
$\ell$ is closest to it.
 Figure~\ref{fig:gadget} depicts the detailed construction of a
 single
gadget. Two consecutive gadgets will be connected as shown in
Figure~\ref{fig:gadgetConnection}. Note that gadget $\ell$ connects
to the $\sink$ (see Figure~\ref{fig:const}). The position of the
$\sink$ and the orientation of the instance depicted in
Figure~\ref{fig:const} indicates that the packets move ``upward."

Given the value of $\ell$, we define the size of each packet to be
$k=\ell!$. We first describe a generic gadget that is used in
constructing each of the $\ell$ gadgets. Figure~\ref{fig:gadget}
depicts the construction of a gadget $i$; the figure shows the
actual construction and a schematic representation, which will be
used subsequently. A gadget is defined by parameters $i$, its gadget
index, and $k$, the capacity of the packets. It consists of $ik$
parallel paths that are disconnected from each other (except for
some special edges called off-ramps described later). Each of these
$ik$ paths consists of $k/i$ nodes; note that $k/i$ is an integer
because $k=\ell!$ and $i \leq \ell$. Therefore, each gadget has
$k^2$ nodes. The two end nodes in each of the paths is designated
either as a head node or the tail node depending on whether it is
closer to or farther away from the $\sink$, respectively.
Furthermore, one of the $ik$ paths is a special path that is called
a ``segment of the shortest path corridor"  and is shown by thick
triple lines in the schematic. When the gadgets are put together to
form the entire instance, these segments will join to form a
sequence of segments from the farthest gadget (away from the
$\sink$) all the way to the $\sink$. This sequence of segments form
the shortest path corridor or \spc.

In each gadget, the node connected to the tail node of the segment
of the \spc~plays a special role; in Figures \ref{fig:gadget} and
\ref{fig:gadgetConnection}, they are depicted as star shaped nodes.
We call them gateway nodes because all packets enter a gadget
through its gateway node. Borrowing from the terminology used in
highways in the United States, the $(i-1)k$ edges coming into the
gateway node from gadget $i-1$ are called on-ramps. There are $ik-1$
edges going from the gateway node to the tails in the gadget (except
for the tail of the segment of the \spc). These edges are called
off-ramps. See Figure \ref{fig:gadgetConnection} for a depiction of
 two consecutive gadgets along with how they are
connected; again, the schematic representation is also provided. To
construct the entire instance, the gadgets are placed one on top of
the other such that their individual segments of the shortest path
corridor align and form the full shortest path corridor that extends
from gadget 1 all the way to gadget $\ell$ and then connects to the
$\sink$. This construction of the entire instance is depicted in
Figure \ref{fig:const}.

\begin{figure}[!htb]
\centering \ifpdf
    \includegraphics[width=130mm]{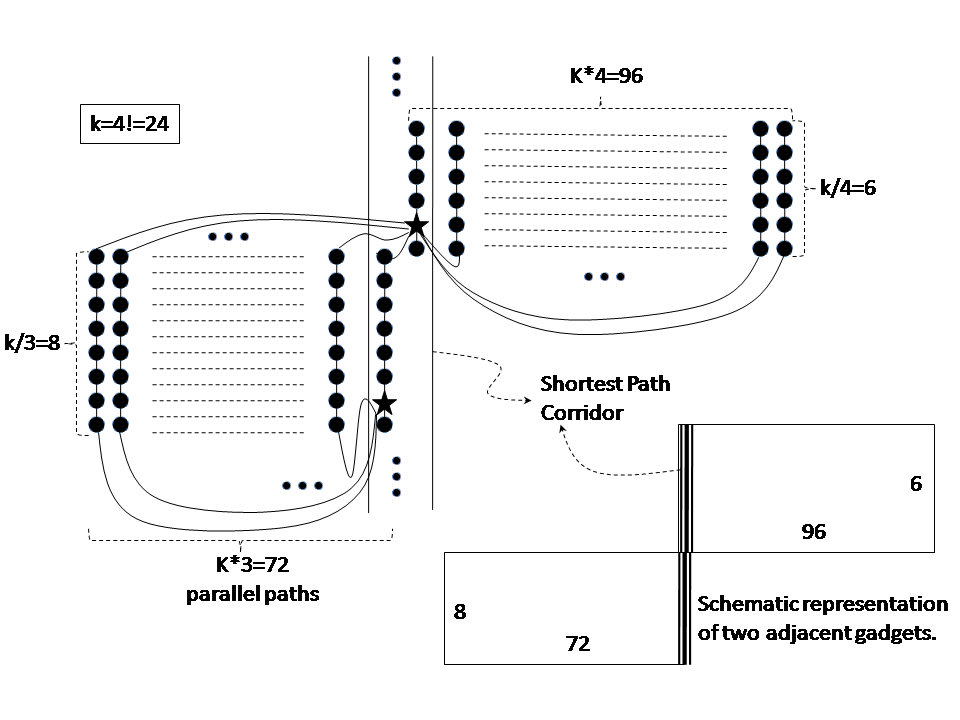}
\else
    \includegraphics[width=130mm, bb=0 0 960 720]{LBGadgetConnection.png}
\fi \caption{Connecting two gadgets in adjacent gadgets. The box
figure on the bottom right is the schematic representation for the
actual construction in the top left.} \label{fig:gadgetConnection}
\end{figure}

\begin{lem}\label{lem:inapproxOPT}
There is a solution to the convergecast problem on the instance
depicted in Figure \ref{fig:const} that hops at most $k^2\ell +
k\ell^2$ times.
\end{lem}
\begin{proof}
The solution works as follows. Each gadget has $k^2$ nodes.
Therefore, gadgets 1 to $i$ have $ik^2$ readings that enter the
gateway of gadget $i+1$.
Then the gateway node, instead of sending them up the \spc,
redistributes these packets to each of the $(i+1)k$ lanes in the
gadget at level $i+1$. Therefore, each lane gets a packet that
contains $\frac{i}{i+1}k$ readings that travel up each lane
collecting the $k/(i+1)$ readings in that lane. Therefore, at the
top of each lane in gadget $i+1$, the number of readings is
$\frac{i}{i+1} k + \frac{k}{i+1} = k$, hence forming a full packet.
These $(i+1)k$ full packets hop into the gateway at gadget $(i+2)$
and proceed toward the $\sink$ in like manner (i.e., avoiding the
\spc~and taking the lanes). Note that at gadget $i$, the following
hop types occur. Firstly, the gateway node at gadget $i$ feeds
$(i-1)k$ packets
  (that it received from gadget $i-1$) via the off-ramps to the
  tail
  nodes in gadget $i$. This takes $ik -1$ hops; although there
  are $ik$ paths, there is no need for
  a hop from the gateway to the segment of the \spc. Secondly, the
  $ik$ packets travel up the lanes costing $k/i$ hops per
  lane. This adds up to $k^2$ hops. Note that this includes the
  on-ramp hops that will carry the packets from gadget $i$ into the
  gateway of gadget $i+1$.
Therefore, at each level $i$, we incur a cost of $k^2 + ik-1$.
Considering this over all $\ell$ levels, the total cost is at most
$k^2 \ell + \sum_{i=1}^{\ell} (ik-1) \leq k^2\ell + k \ell^2$.
\end{proof}

Note that the cost incurred by the solution described in Lemma
\ref{lem:inapproxOPT} hinges on the ability of the gateway nodes to
pack in a non-elementary fashion. Hence it is not elementary in
nature. Also, since it uses the off-ramps, it is not a shortest path
solution either. We shift our concern to solutions that either use
the shortest path or are elementary in nature. The key intuition
here is that such solutions will transmit all the readings entering
the gadget at level $i$ only through the \spc. While the solution in
Lemma \ref{lem:inapproxOPT} was able to split the $k(i-1)$ full
packets into $ik$ partial packets and ride up the gadget (in some
sense, for free), the restricted solution will have to pay for these
packet hops up the \spc. We dissect this cost in Lemma
\ref{lem:inapproxSP} and Lemma~\ref{lem:inapproxElem}. Before that,
we state Lemma~\ref{lem:UniqueSP}, a simple observation about the
instance $\El$.

\begin{lem}\label{lem:UniqueSP}
The tail nodes (except those of the \spc~segments) have exactly two
shortest paths to the $\sink$. All other nodes (including the tail
nodes of \spc~segments) have exactly one shortest path to the
$\sink$.
\end{lem}

\begin{proof}
The tail nodes that are not in the \spc~segments can  go through the
gadget in two ways. They can either go via the off-ramps into the
\spc, or go through the paths for which they are the tails. All
other nodes, it is easy to see, have just one choice.
\end{proof}

The \spt~incurs a higher hop count than the algorithm described in
the proof of Lemma \ref{lem:inapproxOPT}.
Lemmas~\ref{lem:inapproxSP} and \ref{lem:inapproxElem} formalize
this limitation of \spt. The proofs of either lemmas show that their
respective assumptions (namely, shortest path and elementary
packing) force packets to take the \spc, which in turn forces them
to hop at least $2k^2\ell - k^2 \log \ell$ times.

\begin{lem}\label{lem:inapproxSP}
Any shortest path solution to the instance $\El$ depicted in Figure
\ref{fig:const} requires at least $2k^2\ell - k^2 \log \ell$ hops.
This holds regardless of whether the shortest path solution is
deterministic or randomized.
\end{lem}
\begin{proof}
Each gadget produces $k^2$ readings because that many nodes are
present in the gadget at that level. This has two consequences.
Firstly, the number of hops within a gadget, not counting the hops
of packets entering the gadget but counting the off-ramp hops, is at
least $k^2$.
The total number of such hops over all $\ell$ gadgets is $k^2 \ell$.
Secondly, the $k^2$ readings originating in gadget $i$ must each
travel a distance of $(k/(i+1)+k/(i+2)+ \cdots +k/\ell)$, where each
term accounts for the height of gadget $i+1$ up to gadget $\ell$. We
call these the \spc~hops because these readings must travel up the
\spc. Any alternate routing will violate the shortest path property.
Hence, we can argue (in similar lines as in Theorem \ref{thm:uccp})
that any optimal shortest path solution will form $k$ full packets
at the gateway node of gadget $i+1$. Hence, the total number of
packet hops will be $k[(k/(i+1)+k/(i+2)+ \cdots +k/\ell)]$. The
total number of \spc~hops originating over all $\ell$ gadgets is
\begin{eqnarray*}
  k^2 &[&(1/2 + 1/3+1/4+ \cdots + 1/\ell) + (1/3+1/4 + \cdots + 1/\ell) + \cdots + (1/\ell)]\\
   &=& k^2 [(\sum_{i=2}^{\ell} \frac{i-1}{i})] \approxeq k^2 [\ell - \log \ell].
\end{eqnarray*}
Therefore, the total number of hops is at least $k^2 \ell + k^2
[\ell - \log \ell] = 2k^2\ell - k^2 \log \ell$.

We note here that a randomized shortest path solution does not have
much flexibility because of Lemma~\ref{lem:UniqueSP}. The readings
from the tail nodes have two choices. However, any tail node that
takes the off-ramp into the \spc~will contribute to the two types of
hops mentioned regardless of the choice it makes. If it goes through
the \spc, it might contribute to more. Therefore, they are better
off traveling through their individual paths. Hence randomization
does not help in decreasing the number of hops.
\end{proof}

\begin{lem}\label{lem:inapproxElem}
Any  elementary solution to the problem instance $\El$  requires at
least $2k^2\ell - k^2 \log \ell$ hops. This holds regardless of
whether the shortest path solution is deterministic or randomized.
\end{lem}

\begin{proof}
To prove this, all we need is to show that the ``best" elementary
solution will essentially route packets to the $\sink$ in the same
manner as described in Lemma~\ref{lem:inapproxSP}. In other words,
we need to show that all packets entering a gadget through the
gateway node must travel through the \spc~to the $\sink$. The
instance $\El$ is constructed such that only the gateway nodes have
 degree greater than 2. Therefore, to ensure that an algorithm for
$\El$ is elementary, we only need to ensure that gateway nodes
observe the elementary packing property.

Consider the gateway node in gadget $i+1$. The readings routed
through this gateway can be subdivided into those readings that {\em
must} be routed through the gateway and those that have an alternate
route. We first consider the readings that have an alternate route
and show that, for the purpose of analysis, they can be assumed to
take the alternate route rather than through the \spc.  The reading
that have an alternate route are the readings that originate from
nodes in gadget $i+1$ itself, but not in the segment of the \spc~in
that gadget. Consider all the readings from non-\spc~paths in gadget
$i+1$. They form $(i+1)k-1$ paths and each path is of length
$\frac{k}{i+1}$. If these readings moved in the tail-to-head
direction along the path they were in (instead of using the \spc),
they would require exactly $\frac{k}{i+1} ((i+1)k -1)$ hops, which
equals the number of non-\spc~nodes in gadget $i+1$. This implies
that exactly one hop must be accounted for each node's reading.
Since each node requires at least one hop, routing this readings in
any other way will not improve the hop count. Further, this
tail-to-head routing does not violate the elementary packing
principle. Hence, for any elementary solution, we can always
construct another solution in which the readings from nodes not in
the \spc~don't use the segment of the \spc~in their gadget.

The readings that must go through the gateway node are as follows.
\begin{enumerate}
\item It will receive $ik^2$ readings from gadgets 1 through $i$.
\item It has its own reading,  and
\item it also receives 1 reading from the tail node in the segment of the
\spc~in gadget $i+1$.
\end{enumerate}
The elementary packing property therefore requires that exactly 1
partial packet (containing exactly 2 readings) will hop out of the
gateway node. Quite obviously, all the full packets (in any
reasonable elementary algorithm) will follow the \spc. The partial
packet will also move up the \spc~because if it were to take the
off-ramp and go up the gadget through any other path, it will only
incur extra hops.

Now that we have shown that the elementary packing property forces
the routing to be similar to the one shown in
Lemma~\ref{lem:inapproxSP}, we can invoke the mathematical machinery
in that lemma to conclude the proof.
\end{proof}

\begin{thm} \label{thm:inapproxSP}
For any fixed $\epsilon > 0$, there is no
$(2-\epsilon)$-approximation algorithm for \uccp~that follows the
shortest path property. This holds even if randomization is
permitted.
\end{thm}
\begin{proof}
Using the number of hops counted in Lemmas \ref{lem:inapproxOPT} and
\ref{lem:inapproxSP} in the asymptotic sense, the approximation
ratio for any shortest path algorithm is at least

\begin{eqnarray*}
\lim_{\ell \rightarrow \infty} \frac{2k^2\ell - k^2 \log
\ell}{k^2\ell + k\ell^2} &=&\lim_{\ell \rightarrow \infty}
\frac{2k^2(\ell - \frac{\log \ell}{2})}{k^2(\ell + \frac{\ell^2}{k})} \approxeq \lim_{\ell \rightarrow \infty}\frac{2(\ell - \frac{\log \ell}{2})}{l}\\
&=&\lim_{\ell \rightarrow \infty} (2-\frac{\log \ell}{\ell}) = 2.
\quad \quad \quad (\mbox{since } k \gg \ell)
\end{eqnarray*}
Since the limit reaches 2 from below, the theorem holds.
\end{proof}

The following theorem also follows similarly except that we must use
Lemma~\ref{lem:inapproxElem} instead of Lemma~\ref{lem:inapproxSP}.

\begin{thm} \label{thm:inapproxElem}
For any fixed $\epsilon > 0$, there is no
$(2-\epsilon)$-approximation algorithm for \uccp~that respects the
elementary packing property. This holds even if randomization is
permitted.
\end{thm}

\section{\spt~on Tree and Grid Networks} \label{sec:special} We now turn our attention
to the performance of \spt~on special cases based on the graph $G$.

\begin{thm}\label{thm:sptTree}
\spt~is optimal for \uccp~when the underlying graph $G$ is a tree.
\end{thm}
\begin{proof}
Since $G$ is a tree, all the readings from the descendants of any
vertex $v$ (including $v$'s reading) will have to pass through $v$.
Suppose there are $R_v$ such readings. Then any algorithm will have
to transmit at least $\lfloor \frac{R_v}{k}\rfloor + \lceil
\frac{R_v \mod k}{k}\rceil$, which is precisely the number of packet
transmissions out of $v$ in \spt. Therefore, \spt~is optimal with
respect to the number of packet hops.
\end{proof}

 Suppose
the graph $G$ is a grid with  $\vm$ rows  and $\vn$ columns and the
$\sink$ is the vertex at $(1,1)$, i.e. at row 1 and column 1. Since
we are interested in the asymptotic behavior, we assume that $\vm$
and $\vn$ are $\omega(k)$. Furthermore, without loss of generality,
we assume that $\vm$ and $\vn$ are multiples of $k$. We show that
\sptg, an implementation of \spt with a specific underlying shortest
path tree designed for grids, is asymptotically optimal. Whether all
underlying shortest path trees lead to asymptotic optimality remains
open.


The specific shortest path tree for \sptg~on an $(\vm \times \vn)$
grid is as follows: we designate each edge in $G$ to be ``vertical"
(resp. ``horizontal") if it connects vertices from the same column
(resp. same row). All vertical edges are included in the \spt~tree;
horizontal edges are included iff they are from row 1. Intuitively,
the packets move up the columns until they reach the first row. Once
they reach the first row, they move towards the $\sink$ along the
first row. Note that in keeping with our definition of \spt, once a
packet becomes full, it does not split.

\begin{thm}\label{thm:sptGrid}
\sptg~is asymptotically optimal for \uccp~when the underlying graph
$G$ is an $\vm \times \vn$ grid, provided $\vm$ and $\vn$ are in
$\omega(k)$.
\end{thm}

\begin{proof}
We begin by evaluating $h_{LB}$, the lower bound on number of hops
required by any algorithm. Consider a horizontal cut in $G$ betweens
rows $i$ and $i+1$. There are $(\vm - i) \vn$ readings below this
cut. All these readings must pass through this cut. Assume that they
pass through in full packets. Therefore at least $(\vm - i) \vn /k$
hops will pass through the cut. Considering all such horizontal
cuts, the number of hops crossing these cuts must be at least
$\sum_{i=1}^{\vm} (\vm - i) \vn /k = \frac{\vm \vn(\vm-1)}{2k}$.
Similarly, we can also construct vertical cuts which induce at least
$\frac{\vm \vn(\vn-1)}{2k}$ row-wise hops. Therefore, any algorithms
will require at least $h_{LB}$ hops given by
\begin{equation}
h_{LB} = \frac{\vm \vn(\vm-1)}{2k} + \frac{\vm \vn(\vn-1)}{2k} =
\frac{\vm \vn (\vm +\vn -2)}{2k}. \label{eqn:GridLB}
\end{equation}

\begin{figure}[!htb]
\centering \ifpdf
    \includegraphics[width=90mm]{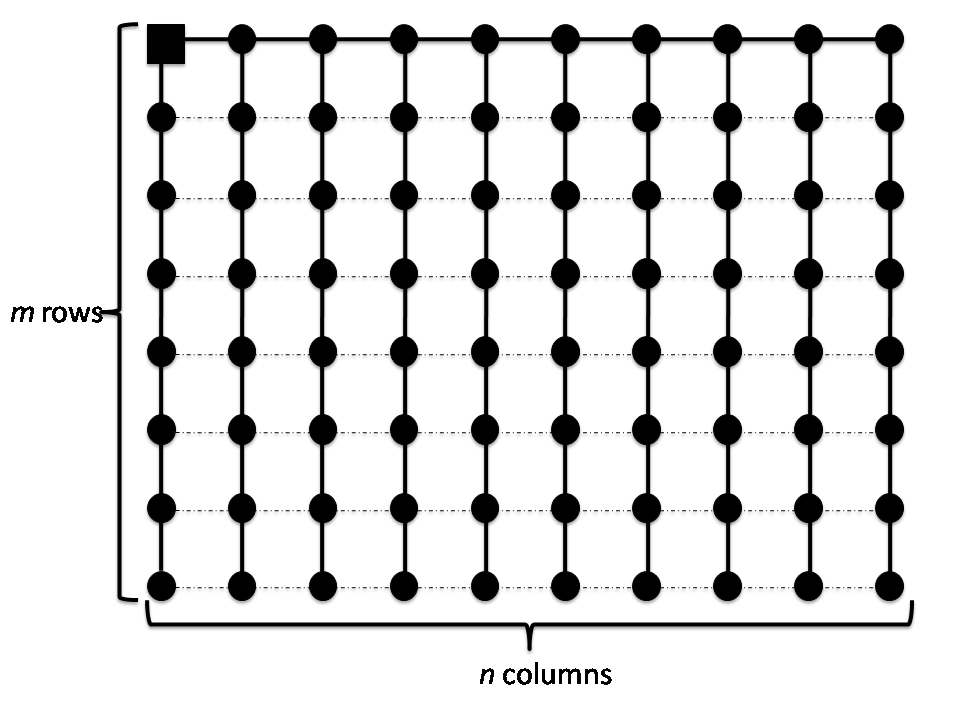}
\else
    \includegraphics[width=90mm, bb=0 0 960 720]{grid.png}
\fi \caption{\sptg~on grid. The square vertex is the $\sink$. The
full edges form the shortest path tree, while the dotted edges are
discarded.} \label{fig:grid}
\end{figure}

\sptg~starts with moving the packets up along columns. Once all the
readings in a column are collected on the first row, the packets
then move along the row to the $\sink$. In each column, as the
packets move upward, a new full packet is formed every $k$ vertices.
If we count all the partial hops in a single column, they are at
most $m-1 \leq m$. Since there are $n$ columns, there are at most
$mn$ partial hops. Since the lower bound on the number of hops (from
Equation \ref{eqn:GridLB}) is $O(\vm \vn (\vm + \vn))$, the partial
hops don't have any bearing on the asymptotic approximation ratio.
Therefore, we are interested in evaluating $h^{\uparrow}$ and
$h^{\leftarrow}$, which are the number of full packet hops up (along
columns) and left (along rows) respectively.

There are at most $\vm / k - 1 \leq \vm/k$ full packets formed in
each column. The first full packet is formed at row $\vm -k$ and
full packets are formed regularly at an interval of $k$ packets.
From the vertex at which a full packet is formed, it will have to
travel up to row 1. Therefore, the number of full packet hops in
each column is at most
\begin{eqnarray} (\vm -k) + (\vm - 2k) + \cdots + (\vm- (\vm/k)k &\leq& (\vm^2/k) - k(1 +2
+\cdots + \vm/k) \notag\\
&\leq& (\vm^2 /k) - k \frac{(\vm/k)^2}{2} \notag \\
&=& \frac{\vm^2}{2k}. \notag
\end{eqnarray}
Since there are $\vn$ columns in total, the number of hops up the
columns, $h^{\uparrow}$, is at most
\begin{equation}
h^{\uparrow} = \frac{\vn \vm^2}{2k}. \label{eqn:GridUp}
\end{equation}
Once the full packets reach the first row, they hop along the row
towards the $\sink$. Each column generates at most $\vm /k$ full
packets. Therefore, the total number of horizontal hops,
$h^\leftarrow$ is given by:
\begin{eqnarray*}
h^\leftarrow &=& (\vm/k) (1 + 2 + \cdots + \vn-1) \notag\\
&=& (\vm/k) (\vn-1)(\vn)/2 \notag \\
&\leq& \frac{\vm \vn^2}{2k}. \label{eqn:GridLeft}
\end{eqnarray*}

Therefore, the total number of full hops for $\sptg$ is at most
\begin{equation}
h^\uparrow + h^\leftarrow = \frac{\vm \vn (\vm + \vn)}{2k}.
\label{eqn:GridUB}
\end{equation}

From Equations \ref{eqn:GridLB} and \ref{eqn:GridUB}, it is clear
that the upper bound and the lower bound converge asymptotically.
\end{proof}

\section{Experimental Study} \label{sec:expt}
The lower bound for \spt~is derived from pathological
  problem instances. It is quite likely that it actually does much
  better in practice. We would like to show this via
  experimentation.  For this purpose, we used Python to implement \spt~ and another
algorithm referred to as {\tt BASIC} later.  {\tt BASIC} constructs
a depth first search tree and then uses batch processing to send
sensor readings along the tree to the $\sink$. {\tt BASIC} has been
used in  many real world sensor network deployment such
as~\cite{PIL+09}.  We also computed the three lower bounds on the
number of hops given a network topology for comparison. They are:
\begin{description}
  \item[LB1:] The number of non-$\sink$ vertices $|V|$,
  \item[LB2:] $\sum_{v \in V} d(v)/k$, where $d(v)$ is the smallest number
  of hops between $v$ and the $\sink$, and
  \item[LB3:] $\sum_{i=1}^{D} n_i/k$, where $D$ is the distance from the
  $\sink$ to the vertex farthest away from the $\sink$ and $n_i$ is the number of vertices
  in $G$ whose distance from the $\sink$ is at least $i$.
\end{description}
 \begin{figure}[!htb]
\centering \ifpdf
    \centerline{\hbox{ \hspace{0.0in}\includegraphics[width=0.33\textwidth]{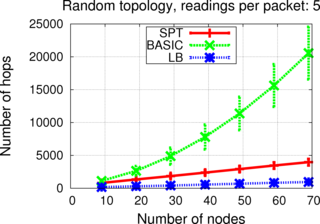}
  \hspace{0in}
  \includegraphics[width=0.33\textwidth]{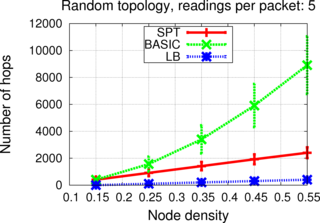}
  \hspace{0in}
  \includegraphics[width=0.33\textwidth]{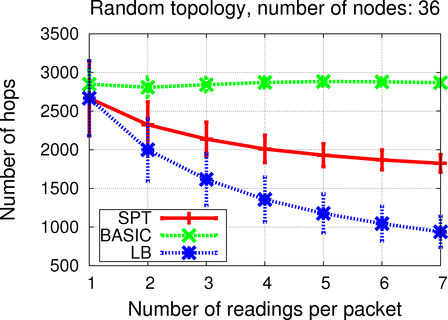}
  }
}
\else
    \centerline{\hbox{ \hspace{0.0in}
  \includegraphics[width=0.33\textwidth, bb=0 0 320 224]{random_size.png}
  \hspace{0in}
  \includegraphics[width=0.33\textwidth, bb=0 0 320 223]{random_density.png}
  \hspace{0in}
  \includegraphics[width=0.33\textwidth, bb=0 0 448 320]{random_k.png}
  }
}
\fi
\caption{Performance of \spt~in random topology: average ratio
of the number of hops for \spt~over the lower bound is 1.45, 1.29,
and 1.14 for the three figures respectively.} \label{fig:random}
\end{figure}

The following parameters are varied to mimic different network
scenarios: network topology,  network size, network density, and $k$
values. 
Figure~\ref{fig:random} shows
the impact of network size, node density and $k$ value
where sensor nodes are uniformly randomly deployed in the network.  We have omitted
the results for grid topologies as similar trends were observed.   In these
figures, we only showed the maximum of the three lower bounds for
comparison. (1) As the network size increases, the number of hops
needed increases almost exponentially when {\tt BASIC} is used  but
only increases moderately when  \spt~is used.  This is because {\tt
BASIC} uses depth first search method to build the collection tree
and the tree height  increases as the network size or increases. (2)
The number of hops for \spt~decreases slightly as density increases,
since the connectedness of the graph increases with density and this
can decrease the depth of the tree. (3) When more readings can be
included in one packet, the number of hops needed decreases when
either algorithm is used, especially for  {\tt BASIC}.  However, the
performance of {\tt BASIC} is still worse than the performance of
\spt~.   The number of hops needed when \spt~is used is only
slightly higher than the lower bound in all cases, which validates
our claim that \spt~works well in practice. In all cases, the ratio
of  the number of hops for \spt~over the lower bound is less than
1.5.


\paragraph{\bf Acknowledgement:} We are  thankful to M. V. Panduranga Rao
and Dilys Thomas for participating in many discussions and providing
valuable suggestions.

\bibliographystyle{plain}
\bibliography{BibTex}

\end{document}